\documentclass[runningheads]{llncs}
\newcommand\fast[2]{#2}

\usepackage{hyperref}
\usepackage{color}
\definecolor{darkgrey176}{RGB}{176,176,176}
\definecolor{darkorange25512714}{RGB}{255,127,14}
\definecolor{lightgrey204}{RGB}{204,204,204}
\definecolor{steelblue31119180}{RGB}{31,119,180}

\usepackage{amssymb,amsmath}
\usepackage{cleveref}

\usepackage{subfig}
\usepackage[T1]{fontenc}
\usepackage[ruled]{algorithm2e} 
\usepackage[normalem]{ulem}

\SetAlFnt{\small}
\SetAlCapFnt{\small}
\SetAlCapNameFnt{\small}
\SetAlCapHSkip{0pt}
\IncMargin{-\parindent}

\newcommand\gammacex{\gamma_{\text{CEX}}}

\usepackage{longtable}
\usepackage{pgfplots}
\usepackage{tikz}
\usetikzlibrary{arrows.meta,positioning,calc,shapes.geometric,fit}
\usepackage{xcolor}
\newcommand{\grayzero}[1]{\textcolor{gray!90!black}{#1}}

\usepackage{makecell} 
\usepackage{enumitem}
\usetikzlibrary{positioning}
\usepackage{multirow}
\newtheorem{assumption}[theorem]{Assumption}
\usepackage{comment}

\begin{document}
\title{Where Does MEV Really Come From? \\ 
Revisiting CEX–DEX Arbitrage on Ethereum}
\titlerunning{Revisiting CEX–DEX Arbitrage on Ethereum}

\author{
Bence Ladóczki\inst{1}\orcidID{0000-0002-5056-4184} 
\and Miklós Rásonyi\inst{2}\orcidID{0000-0002-3105-4752} 
\and János Tapolcai\inst{3}\orcidID{0000-0002-3512-9504}
}

\institute{
BME and HUN-REN Information Systems Research Group, Budapest, Hungary\and
HUN-REN Alfréd Rényi Institute of Mathematics and E\"otv\"os Lor\'and University, Budapest, Hungary, 
\email{rasonyi.miklos@renyi.hu} \and
Budapest University of Technology and Economics (BME), Budapest, Hungary\footnote{Dept. of Telecommunications and Artificial Intelligence, Faculty of Electrical Engineering and Informatics,  
Műegyetem rkp. 3., H-1111 Budapest, Hungary.
\email{ladoczki.bence@vik.bme.hu}, \email{tapolcai@tmit.bme.hu}}
}

\maketitle

\begin{abstract} 
A central question of the Ethereum ecosystem is where Maximal Extractable Value (MEV) revenue originates and to what extent it stems from harming unsuspecting users. 
It is acceptable if MEV arises from arbitrages between centralised and decentralised exchanges (CEX--DEX). 
Yet theoretical models have significantly underestimated the scale of these arbitrages, while empirical studies have highlighted their importance -- though these remain conservative estimates, constrained by numerous debatable heuristic assumptions.
Revisiting the theoretical model, we found that CEX–DEX arbitrages require trading volumes on the order of the total activity of major liquidity pools and yield profits comparable to MEV.
Most prior AMM models utilised the Black–Scholes (BS) stochastic differential equation (SDE) -- i.e., geometric Brownian motion -- and assumed continuous price trajectories where asset prices move in small increments only. 
We argue that BS underestimates arbitrage profits by ignoring price jumps, which are precisely the points at which arbitrage opportunities tend to arise.
To address this gap, we present an extended discrete-time AMM model in which the price process is the sum of a diffusive component and stochastic jumps that can have arbitrary noise distributions. 
Although mathematically more involved this framework allows us to employ a general discrete-time SDE and compute the stationary probability distribution via function iteration with geometric convergence. 
We further prove that the resulting mispricing process is an ergodic Markov chain.
We implement our model in C++, collect spot prices and AMM exchange data from the Ethereum blockchain and fit the model parameters to the observed prices. 
The estimates derived from our model closely match empirical observations and provide a natural theoretical explanation for several fundamental questions in the blockchain ecosystem.
\end{abstract}


\section{Introduction}
In the current implementation of the proof of stake (PoS) consensus protocol of Ethereum each slot is assigned to a single validator to propose a block, but in practice validators often resell this right at surprisingly high prices to so-called block builders. 
Therefore, block building and proposing are distinct processes in PoS Ethereum, hence the name proposal–builder separation (PBS). 
Today, one of the central questions in the Ethereum ecosystem is why block builders are willing to pay so much for this right, and more importantly, where this revenue originates and how ‘dark’ it is, that is, to what extent it derives from harming unsuspecting users. 
MEV (Maximal Extractable Value) denotes the profit that can be captured by including, excluding, or reordering transactions within a block.
The most commonly observed forms of MEV are sandwich and liquidations attacks and arbitrage trades. 
Among these, only arbitrage is considered benign. It is regarded as a necessary mechanism for automated market makers (AMMs)~\cite{hanson2007logarithmic} to keep exchange rates aligned with the external market.

Initially, MEV was believed to consist primarily of sandwich attacks, atomic (DEX–DEX) arbitrage, and back-running attacks -- strategies that are relatively easy to detect on-chain and often yield conspicuously large profits. 
However, as these events are rare, their aggregate contribution is limited; see rows 3 and 4 of Table~\ref{table:comp}. 
In contrast, in the case of a CEX–DEX arbitrage only the DEX leg is on-chain, making it non-trivial to distinguish genuine arbitrage from noise trades. 
This perception aligned with theoretical models~\cite{amm1,amm3,nezlobin2025loss}, which predicted only modest CEX–DEX arbitrage profits and volumes (see row 5 of Table~\ref{table:comp}). 
More recent empirical studies have demonstrated that CEX–DEX arbitrage plays a much larger role in MEV than previously assumed, with transaction volumes exceeding those of atomic arbitrage and sandwich attacks by several orders of magnitude~\cite{yang2025decentralization,heimbach2024non,oz2024wins} (see rows 2–4 of Table~\ref{table:comp}).
These results suggest that the true scale of CEX–DEX arbitrage is substantially greater than predicted by existing theoretical models. With this work, we aim to correct and explain this discrepancy.

\begin{table}[h]
\centering
\begin{footnotesize}
\begin{tabular}{r|l|l|cr|cr|c|}
\cline{2-8}
& &A Uniswap V2 pool in Sept. 2025 & \multicolumn{2}{c|}{USDT $\to$ WETH} & \multicolumn{2}{c|}{WETH $\to$ USDT}  & MEV
 \\
\cline{4-7}
& & \tiny{\texttt{0x0d4a11d5eeaac28ec3f61d100daf4d40471f1852}} & \#swaps & volume & \#swaps & volume & /profit \\
\cline{2-8}
\tiny{1} &\multirow{4}{*}{\rotatebox{90}{Empirical\hspace{1mm}}} & All swap transactions   & 2724 & \$3\,369\,073
& 1907 & \$3\,072\,624  &
$\sim\$2\,030$ \\
\tiny{2} &&\quad CEX-DEX arb. (heur. in \cite{wu2025measuring})
& 496  & \$970\,436 & 241  & \$1\,002\,020 & 
\\
\tiny{3} &&\quad Sandwich (avg. in \cite{wu2025measuring}) & 0.2 & \$25\,345 & 0.2 & \$64\,443 & $\$414$\\
\tiny{4} &&\quad Atomic arb. (avg. \cite{wu2025measuring}) &  3 & \$10\,515 & 2 & \$34\,562 & $\$136$\\ 
\cline{2-8}
\tiny{5} &\multirow{2}{*}{\rotatebox{90}{SDE\hspace{1mm}}} &Black–Scholes (LVR~\cite{amm1,amm3,nezlobin2025loss}) 
& $86$& \$$208\,844$ &$75$& \$$177\,961$ &\$$48$ \\ 
\tiny{6} &&Our model 
&$256$& \$$1\,881\,277$ &$251$& \$$1\,838\,985$ &\$$1\,454$ \\ 
\cline{2-8}
\end{tabular}
\end{footnotesize}
\caption{Average daily swaps and volumes, comparing all transactions with arbitrage transactions between $13-29$ September 2025 in one of the largest Uniswap V2 pools by TVL (\$$86$M).
For context, the total DEX trading volume during this period was roughly \$$462$B, with Uniswap alone accounting for \$$106$B. 
The upper part of the table reports empirical measurements, while the lower part presents theoretical estimates based on various stochastic differential equation (SDE) models for CEX–DEX arbitrage. 
Our model explicitly involves price jumps, with parameters fitted from the exchange rate time series (see \Cref{sec_fitting}). 
Empirical MEV was estimated by considering only single-swap transactions (ignoring bundles), summing their priority fees and direct coinbase transfers, and then scaling the result proportionally to account for bundle transactions.
}
\label{table:comp}
\end{table}

Table~\ref{table:comp} compares empirical measurements 
 with theoretical estimates of swap volumes and counts for a Uniswap V2 ETH–USDT pool (one of the largest by total value locked, TVL).
The theoretical estimates are derived from a stochastic model describing the price process of AMMs and the arbitrageurs' response to CEX-DEX mispricing. 
Since such models can only be applied to a specific trading pair, the empirical data from~\cite{wu2025measuring} was filtered to this pool.
Roughly speaking, the theoretical estimates represent the expected frequency and volume of arbitrage transactions implied by the SDE on an infinite-horizon. 
We argue that accurate stochastic models should yield estimates consistent with empirical observations.

Most previous theoretical work tackled the problem with simplified tools and derived closed-form equations to describe losses in AMMs and arbitrage profits~\cite{amm1,amm3,nezlobin2025loss}. 
The trade-off for this computational simplicity lies in that these AMM models are essentially all based on the assumption that price trajectories are continuous and that the spot price of the asset under consideration can only change by a small amount in a small interval. 
The backbone of these frameworks is the Black–Scholes (BS) price model developed in the 1970s. 
While BS is a commonly used workhorse in financial mathematics, many believe that arbitrage activity tends to emerge precisely at price jumps. 
However, incorporating price jumps into AMM models requires more sophisticated approaches than a two-parameter SDE can provide.

The main contribution of this paper is an extended discrete-time AMM model that integrates both a diffusive component and stochastic jumps, while accommodating arbitrary noise distributions. 
Specifically, we employ a discrete-time general SDE and show that the stationary probability distribution (SPD) can be computed via function iteration that converges with geometric speed. 
We prove that the state space of the process under consideration is a small set, which implies that the misprice process is an ergodic Markov chain. 
Our extension can be reduced to the BS model, the Merton jump diffusion SDE~\cite{merton} or the Ornstein-Uhlenbeck SDE~\cite{uomotion} when the model parameters are chosen accordingly.

After calibrating the model parameters to match the statistical properties of real-life data, the SDE produces estimates similar to the observed values. 
In essence, by determining seven parameters and running multiple rounds of numerical integration, we obtained an estimated arbitrage volume that was roughly twice than that observed on-chain.
Note that our model (as well as prior ones~\cite{amm1,amm3,nezlobin2025loss}) excludes noise traders and assumes the arbitrageur trades with high probability whenever the CEX–DEX price gap is sufficiently large. 
In other words, a significant proportion of the 24-hour trading volume appears to be driven by the pool’s need to track the CEX price\footnote{This implies that the 1-day vol/TVL metric greatly depends on the exchange-rate dynamics of the token pair and the DEX fee. Interestingly, for Uniswap V2 pools between a given currency and various stablecoins, this metric takes similar values regardless of the pool’s TVL (see \url{https://app.uniswap.org/explore/pools}).}. Thus, in terms of volume, noise traders—agents who swap irrespective of the CEX price—play only a minor role (they are likewise omitted in the models of~\cite{amm1,amm3,nezlobin2025loss})

In empirical studies, swaps classified as CEX–DEX arbitrage by strict heuristics, e.g. MEV bots run by searchers. 
These heuristics exclude transactions seen in the public mempool, require that the trade be the first swap in its direction within a given pool and a block, and filter out atomic MEV strategies such as sandwiches, DEX-DEX arbitrage and liquidation attacks. 
Additional filters remove order flow auction (OFA) backruns, transactions routed through known smart contracts or labeled bots, and swaps involving NFTs~\cite{heimbach2024non,wu2025measuring}.

Our theoretical results suggest that a significant portion of the remaining swap volume is also executed by (hidden) arbitrageurs.
However, when looking at transaction counts, the observed number of swaps is higher than predicted by our theoretical model, which assumes that each arbitrage opportunity is exploited by a single arbitrageur.
Note that the specific actor may differ across instances; from the perspective of the model, the identity of the arbitrageur is irrelevant.
This discrepancy can be explained by the presence of many arbitrageurs in the network--each utilising slightly different CEX prices--as well as by the widening of the CEX bid–ask spread following price jumps. As a result, the CEX–DEX price gap is sometimes closed not by a single transaction, but through a small sequence of transactions in practice.

The paper is organised as follows. 
First, Section~\ref{sec:background} provides a brief overview of related work, then the considered price dynamics is discussed.
Next, in Section~\ref{sec:invariantprobability} the invariant probability distribution of the misprice process is given in terms of a function iteration formula. 
A numerical solver to determine the invariant distribution is described in Section~\ref{sec:numericalresults}, where we also present measurements for cross-validation of our findings.
We conclude our work in Section~\ref{sec:conclusions}.

\section{Stochastic Price Models}
\label{sec:background}

First, let us review some of the SDEs utilised in previous works. 
~\cite{g3mnassib} derives its result relying on an $n$-dimensional Brownian motion, i.e. $\frac{dS}{S} = \mu_tdt+dW_t$ for $n$ assets, using the usual definition for the model parameters.
~\cite{amm1} assumes that the CEX price is observable, and evolves exogenously according to a geometric Brownian motion (GBM) that is a continuous $Q$-martingale. Their SDE is $\frac{dP_t}{P_t} = \sigma_t dB_t^{\mathbb{Q}}$, with $B_t^{\mathbb{Q}}$ being the Brownian motion. 
In financial mathematics, a $Q$-martingale often refers to processes defined under the risk-neutral measure ($Q$)~\cite{delbaen2006mathematics} and its future expected value equals its current value, given the past information. 
The case without a price trend is particularly convenient (notice that $\mu_t$ is missing from the aforementioned SDE), as it allows for the analysis of the system's asymptotic behavior.

Studies ~\cite{Cartea2022Decentralised,Cartea2023PredictableLosses,amm3} abandon the martingale assumption and extends the model by using the SDE $\frac{dP_t}{P_t} = \mu dt + \sigma dB_t$ to allow for the presence of price trends ($\mu$), deriving formulas for this more general case.
Their SDE is not governed by a driftless geometric Brownian motion, but by the full Black–Scholes model~\cite{blackscholes}.
It is important to note that the Black–Scholes model has been a cornerstone of quantitative finance since the 1970s.
However, more sophisticated models have since emerged, offering better alignment with observed market behaviour.
In particular, the Black–Scholes model fails to capture several important stylised facts about log-price processes, such as the presence of heavy tails~\cite{cont}. 

As emphasised in~\cite{merton}, the Black-Scholes model is only valid when trading takes place in continuous time and the price process has continuous paths. 
In other words, the exchange rate has a "local" Markov property, namely that it cannot change significantly in a short time interval. 
The conjecture of the current work is that the assumption that trading takes place in continuous time and that the price trajectories are continuous is misleading. 
To underpin our suspicion, we analysed two months of trading data provided by the Binance Spot price API between November and December 2024 and found that in some cases, the amplitude of price jumps in the ETH price can be two orders of magnitude larger than the instantaneous volatility. 
For a concrete example, we point out the jump that occurred at 2024-12-09 22:06:29, when the spot price changed from $3542$ to $3640$ in a second. 
Note that this was not the only occurrence of a huge jump in the spot price of ETH.
This clearly indicates that prior works relying on the continuous path assumption cannot accurately describe the interaction between the arbitrageurs and the DEXs. 
In this work, we fill this research gap and investigate the effects of discontinuous price trajectories. 
In a somewhat similar fashion to our objectives, Dewey et al. in~\cite{Dewey2023CFMM} use a Merton jump diffusion SDE in the following form: $ \frac{dp_t}{p_t} = (\mu_D - \lambda_J k) \, dt + \sigma_D \, dB_t + (y_t - 1) \, dN_t $ (refer to the original paper for the definitions) and investigate AMM dynamics with this more comprehensive price model.
Regarding the generality of the proposed mathematical framework, our model covers the two-parameter BS model, the Merton jump diffusion SDE and the widely used mean-reverting Ornstein-Uhlenbeck~\cite{uomotion} SDE as well.

\subsection{Price dynamics}\label{mathematicalmodel}
\label{sec:pricedynamics}
\begin{table}[ht]
\begin{small}
\centering
\begin{minipage}{0.45\textwidth}
\begin{tabular}{| l | l |}
\hline
\textbf{Notation} & \textbf{Explanation} \\
\hline
$X_n$ & CEX price at time $n$ \\
$\sigma$, $\mu$ & Volatility and drift \\
$\epsilon$ & Noise \\
$Z_n$ & Jump indicator\\
$U_n$ & Jump amplitude \\
$q$  & Jump probability \\
$p$  & Arbitrage trade probability \\
\hline
\end{tabular}
\end{minipage}%
\hspace{1cm}
\begin{minipage}{0.45\textwidth}
\begin{tabular}{| l | l |}
\hline
\textbf{Notation} & \textbf{Explanation} \\
\hline
$\hat{X_n}$ & Auxillary process \\
$\tilde{X_n}$ & AMM dynamics \\
$\gamma_{-},\gamma_{+}$ & Fee levels \\
$u()$ & Jump amplitude distribution \\
$f()$ & Noise density\\
$f_k()$ & $k$-th iterated density \\
$f^{*}()$ & Invariant density \\
\hline
\end{tabular}
\end{minipage}
\end{small}
\caption{Notations used in this work. }
\label{tab:notations}
\end{table}
Next, we introduce the equations to be then used to describe and analyse arbitrage opportunities and arbitrageur profits while considering a large class of price models. Our notational convention is presented in Table~\ref{tab:notations}. 
We choose seconds as a unit to measure time, mostly because the block time of Ethereum is 12s, therefore it is more convenient to convert all the model parameters to seconds and use them accordingly. The results in the tables are converted back to daily values.

First, let us set the conditions of our mathematical model. 
At time $n\in\mathbb{N}$ the exchange rate of a risky asset (e.g. ETH), against a stablecoin (e.g. USDT or USDC ) is $P_{n}>0$ at a CEX. We define $X_n:=\ln(P_n)$.
The DEX price is usually different from this due to the passive nature of traditional on-chain AMMs. 

At time $n$ it is expressed as $\tilde{P}_{n}>0$. 
The (logarithmic) mispricing~\cite{amm3} at time $n$ is then defined as $\tilde{X}_{n}:=\ln(P_{n}/\tilde{P}_{n})$, $n\in\mathbb{N}$. 
We treat $X_n,\tilde{X}_n$ as a stochastic processes, i.e. $X_{n}$, $n\in\mathbb{N}$ are a sequence of random variables on a fixed probability space $(\Omega,\mathcal{F},P)$ with the usual definitions. 
Expectation of a random variable $Y$ will be denoted as $E[Y]$.

\begin{remark}{\rm 
Throughout this work we work under the objective (statistical) probability $P$, and \emph{not} under a risk-neutral probability. 
Risk-neutral measures are convenient tools for derivative pricing, but they are artifacts without statistical meaning, and hence not suitable for empirical analysis based on exchange rate data. 
Moreover, risk-neutral probabilities are typically well-defined only on finite horizons, whereas our results concern asymptotic properties (e.g., the limiting distribution of mispricing) as the time horizon tends to infinity. 
In this context, the objective probability is the only meaningful choice.
}
\end{remark}

\subsection{Recursive Description}
We now describe the random movement of the CEX price process $X_n$, $n\in\mathbb{N}$.
We choose to describe our stochastic processes in a discrete fashion, because computer simulations proceed in discrete time steps anyway, so the following recursive description is employed:
\begin{equation}\label{evolve}
X_{n+1}-X_{n}=\mu(X_{n})+\sigma(X_{n})\varepsilon_{n+1}+Z_{n+1}U_{n+1},\ n\in\mathbb{N}.
\end{equation}
In other words, the CEX price process is governed by both diffusive and jump terms. 
The first two terms on the right-hand side of \eqref{evolve} are the \emph{diffusive} terms:
the function $\mu:\mathbb{R}\to\mathbb{R}$ is the \emph{drift} and $\sigma:\mathbb{R}\to\mathbb{R}_{+}$ is the \emph{volatility}. 
$\varepsilon_{n}$, $n\in\mathbb{N}$ is an i.i.d. noise term.
We assume that the functions $\mu,\sigma$ are (Borel)-measurable and the common law of the $\varepsilon_{n}$, $n\in\mathbb{N}$ is absolutely continuous with density function $f(x)$. 

The third term on the right-hand side of \eqref{evolve} represents random \emph{jumps}.
We assume that $Z_{n}$ is an i.i.d. $\{0,1\}$-valued sequence, indicating whether a jump occurs at time $n$ (that is, $Z_{n}=1$) or not ($Z_{n}=0$). 
We assume that $q:=P(Z_{0}=1)<1$.
The \emph{jump sizes} $U_{n}$, $n\in\mathbb{N}$ are i.i.d. with absolutely continuous law and the corresponding density function is $u(x)$. 

The initial value $X_{0}$ can be either deterministic or random, but we always assume that $X_{0}$, $(\varepsilon_{n})_{n\in\mathbb{N}}$, $(Z_{n})_{n\in\mathbb{N}}$ and $(U_{n})_{n\in\mathbb{N}}$ are independent. 
This immediately guarantees that $X_{n}$ is a homogeneous Markov chain: from the current state $X_{n}=x$ the next state $X_{n+1}$ is calculated as $x+\mu(x)+\sigma(x)\varepsilon_{n+1}+Z_{n+1}U_{n+1}$, using the i.i.d. noise sequence $(\varepsilon_{n},Z_{n},U_{n})$ and no information is carried over from the past.

\begin{remark}{\rm The interpretation of \eqref{evolve} is simple: the process has small, ``diffusive'' movements, but unexpected jumps are also allowed to occur. 
We believe that this extension enables one to describe highly volatile cryptocurrencies more accurately, see for example~\cite{Dewey2023CFMM}.
The term $Z_{n+1}U_{n+1}$ is the discrete-time analogue of a compound Poisson process (or, more generally, of a L\'evy process). 
Indeed, let $V$ be the first time a jump occurs (that is, the first $n$ such that $Z_{n}=1$). 
Clearly, the random variable $V$ has a geometric distribution, which is the discrete-time analogue of the exponential distribution that appears between two consecutive jump times of a Poisson process.}
\end{remark}

\begin{remark}{\rm The model \eqref{evolve} is the discrete-time counterpart of the continuous-time jump-diffusion model (considered e.g.\ in \cite{merton,hanson}) which satisfies the following stochastic integral equation:
\begin{equation}\label{diffusion}
X_{t}=X_{0}+\int_{0}^{t}\mu(X_{s})\, ds+\int_{0}^{t}\sigma(X_{s})dW_{s}+\sum_{j=1}^{N_{t}}\xi_{j},
\end{equation}
where $W_{t}$ is the Brownian motion, $N_{t}$ is a Poisson process and $\xi_{j}$ are i.i.d.\ jump amplitudes. 
$(W_t)_{t\in\mathbb{R}_+},(N_t)_{t\in\mathbb{R}_+},(\xi_{n})_{n\in\mathbb{N}}$ are independent. 
Note, however, that the discretisation of \eqref{diffusion} would result in Gaussian $\varepsilon_{n}$ while in \eqref{evolve} we do not restrict the noise to be of Gaussian distribution.}
\end{remark}

\begin{example}\label{exo}
{\rm When $\mu,\sigma$ are constants, $\varepsilon_{i}$ are standard Gaussian, $P(Z_{0}=0)=1$ (that is, no jumps) then we obtain a \emph{random walk} model~\cite{Bachelier1900}, the discrete-time version of geometric Brownian motion that was meticulously analysed in ~\cite{amm3}. 

When $\mu(x):=-\alpha(x-\theta)$ for some $0<\alpha<2$ and $\theta\in\mathbb{R}$, $\sigma$ is constant and $P(Z_{0}=0)=1$, we are working with an \emph{autoregressive} model, the discrete-time analogue of an Ornstein-Uhlenbeck process~\cite{uomotion}.
One can use more complex $\mu,\sigma$ functions as it usually happens in the field of quantitative finance, see the monograph \cite{ek}. 
We stress once more that $\varepsilon_{n}$ may be non-Gaussian. 
In fact, empirical evidence shows that $\varepsilon_{n}$ has (much) fatter tails than a regular Gaussian distribution, see \cite{cont}, hence realistic models \emph{are} non-Gaussian. 
In Section \ref{sec:statisticaltests} we will evaluate the logarithmic returns of several cryptocurrency pairs and perform statistical tests to demonstrate that the distribution of $\varepsilon_{n}$ is far from Gaussian when considering real-life exchange rates.

If $0<P(Z_{n}=1)<1$ then there is a nontrivial jump component in the discrete price process. 
We do not consider the case $P(Z_n=1)$ (when there is always a jump), which slightly simplifies the ensuing proof.

In continuous time, such models (more generally, L\'evy-process based models) enjoy popularity,
see \cite{cont-tankov,ek}. 
Adding jumps is a possible way of reconciling empirical evidence with theoretical models in quantitative finance, consult~\cite{cont} again and refer to Table~\ref{table:hourlylogreturns} below for statistical tests on cryptocurrencies. }
\end{example} 

We will now shift our focus to examining the effects of arbitrageurs' presence in the market.
Let $\gamma_{+},\gamma_{-}>0$ denote the exchange fee in units of log price for buying and selling tokens in the AMM following the notations of~\cite{amm3}.
The arbitrageur trades against the AMM as follows. 
At each moment $n$, a myopic arbitrageur (maximising short-term profits) arrives at the AMM with probability $0<p\leq 1$. 
When the logarithm of the misprice is outside $[-\gamma_{-},\gamma_{+}]$ then the arbitrageur sets it back to the closest boundary point (that is, either to $\gamma_{+}$ or to $-\gamma_{-}$) thus making
a profit. 
Otherwise no action is taken. 
Note that when there is no arbitrage activity against the AMM, the misprice process is identical to the price process. 
Consequently, the logarithm of the misprice continues to evolve according to~\eqref{evolve}.

Rigorously this can be formulated as follows: let $A_{n}$, $n\in\mathbb{N}$ be an i.i.d.\ sequence of $\{0,1\}$-valued random variables with $0<P(A_{0}=1)=p\leq 1$. 
We assume that $X_{0}$, $U_{n}$, $\varepsilon_{n}$, $Z_{n}$, $A_{n}$ are independent. 
$A_{n}=1$ means that there is an arbitrageur arriving at discrete time $n$. 
Let us introduce the stochastic process $\tilde{X}_{n}$, $n\in\mathbb{N}$ to describe the dynamics of the mispricing between the CEX and DEX prices, with arbitrageurs taken into account.
We fix $\tilde{X}_{0}$ (for instance $\tilde{X}_0=0$) and recursively define, for $n\in\mathbb{N}$,
\begin{eqnarray}\label{tildex}
\hat{X}_{n+1} &:=& 1_{\{\tilde{X}_{n}>\gamma_{+}\}}A_{n+1}\gamma_{+}
-1_{\{\tilde{X}_{n}<-\gamma_{-}\}}A_{n+1}\gamma_{-}\\ \nonumber
&+& (1-A_{n+1})\tilde{X}_{n}1_{\{\tilde{X}_{n}\notin [-\gamma_{-},\gamma_{+}]\}}+
\tilde{X}_{n}1_{\{\tilde{X}_{n}\in [-\gamma_{-},\gamma_{+}]\}},\\  \label{tildey}
\tilde{X}_{n+1} & :=& \hat{X}_{n} + \mu({X}_{n})+\sigma({X}_{n})\varepsilon_{n+1}+Z_{n+1}U_{n+1}.
\end{eqnarray}
The expressions above can be explained as follows: when an arbitrageur arrives at time $n$ and the current price $\tilde{X}_n$ is above $\gamma_{+}$ then the arbitrageur sets it back to $\gamma_{+}$; if the current price is below $-\gamma_{-}$ then the price gets set back to $-\gamma_{-}$; if no arbitrageur arrives then no action is taken even if the process $\tilde{X}_{n}$ is outside $[-\gamma_{-},\gamma_{+}]$. 
When the process is \emph{inside} $[-\gamma_{-},\gamma_{+}]$ then nothing happens no matter what. 
In such a way we arrive at the value $\hat{X}_{n}$ as a result of arbitrage trades. 
After this, the mispricing process evolves following the stochastic dynamics of the CEX price, \eqref{evolve}. 
That is, $\tilde{X}_{n+1}$, the value of the new misprice at time $n+1$, can be calculated from $\hat{X}_n$ in the same way as $X_{n+1}$ was calculated from $X_n$ in \eqref{evolve}.

One can readily check that when $\{A_{1}=0,\ldots,A_{n}=0\}$ and $X_0=\tilde{X}_0$, one has $X_{k}=\tilde{X}_{k}$, $0\leq k\leq n$, that is, without the intervention of arbitrageurs \eqref{tildex} reduces to the original CEX price evolution \eqref{evolve}. 
Owing to the fact $(X_n,\tilde{X}_{n})$ is being generated from an i.i.d.\ noise sequence $(A_n,Z_n,U_n,\varepsilon_n)$ with a recursion, $(X_n,\tilde{X}_{n})$ is, again, a (homogeneous) Markov chain.

\section{Invariant Probability Distribution}
\label{sec:invariantprobability}

Now we make some technical assumptions about various components of the mispricing model 
$\tilde{X}_n$ described in the previous section.

\begin{assumption}\label{bou}
$\mu(x)=\mu\in\mathbb{R}$ and $\sigma(x)=\sigma>0$ are constants.
\end{assumption}

This assumption guarantees that the incoming noise always enters the system with a non-vanishing coefficient 
$\sigma$, in other words, there is enough randomness to make the system ergodic. Under this Assumption,
$\tilde{X}$ itself is a Markov process, as easily seen.

Finally, some regularity of the density $f$ is required. 
The latter assumption holds, in particular, if $f$ is continuous and positive everywhere, which is, again, a mild requirement that is fulfilled in all imaginable models in practice.

\begin{assumption}\label{fou} Let $K:=|\mu|+\sigma$.
Let $H:=(2\max\{\gamma_{+},\gamma_{-}\}+K)/c_{0}$. The density $f$ is bounded away from $0$ in the following sense:
$$
c_{1}:=\inf_{|x|\leq H}f(x)>0.
$$
\end{assumption}

Now let us calculate the conditional density of the CEX price increments, expressed as the convolution of a rescaled density and the corresponding jump term. 

\begin{lemma}\label{density1}
The law of $X_{n+1}-X_{n}$ conditionally to $X_{n}=x$ is absolutely continuous with (conditional) density
$$
h(x,y):=(1-q)\frac{1}{\sigma}f\left(\frac{y-\mu}{\sigma}\right)+{}
q\int_{\mathbb{R}}u(y-z)\frac{1}{\sigma}f\left(\frac{z-\mu}{\sigma}\right)\, dz,
$$	
that is, for all Borel sets $A$,
$$
P(X_{n+1}-X_n\in A|X_n=x)=\int_A h(x,y)\, dy.
$$
\end{lemma}
Next, we turn our attention to the Markov process $\tilde{X}_n$ and determine its transition density.

\begin{lemma}\label{density2}
For each Borel set $A$,
$$
P(\tilde{X}_{n+1}\in A|\tilde{X}_n=x)=
\int_A q(x,y)\, dy,
$$
where the transition density $q(x,y)$ is given by
$$
q(x,y)=h(x,y-x)(1-p)+h(\tilde{x},y-\tilde{x})p.
$$
Here we use the notation
$$
\tilde{x}:=x1_{\{x\in [-\gamma_{-},\gamma_{+}]\}}+\gamma_{+}1_{\{x>\gamma_{+}\}}-\gamma_{-}1_{\{x<-\gamma_{-}\}}.
$$
\end{lemma}

Let $\mu_0$ denote the law of $\tilde{X}_0$. For instance, if $\tilde{X}_0$ is a constant $w$ then $\mu_0=\delta_w$, the probability concentrated on the one-point set $\{w\}$.

\begin{corollary}\label{cor} Whatever $\tilde{X}_{0}$ is, the random variables $\tilde{X}_{k}$, $k\geq 1$ have laws that admit density functions
$f_{k}(x)$ given by
$$
f_1(x)=\int_{\mathbb{R}}q(x_0,x)\mu_0(dx_0)
$$
and by the recursion
\begin{equation}
f_{k+1}(x)=\int_{\mathbb{R}}q(x_{0},x)f_{k}(x_{0})\, dx_{0},\ k\geq 1.
\label{eq:functioniteration}
\end{equation}
\end{corollary}

Now we are in a position to state our main theoretical result: for the misprice process in our model a SPD \textbf{exists}, the system under consideration tends to this invariant state from an arbitrary initialisation (at a geometric speed) and the law of large numbers also holds.

\begin{theorem}\label{main} Denote by $f_{n}$ the density function of $\tilde{X}_{n}$ for $n\geq 1$.
Let Assumptions \ref{bou} and \ref{fou} be in force. Then there is a density function $f^{*}$ such that
such that 
\begin{equation}\label{esti}
\int_{\mathbb{R}}|f^{*}(t)-f_{n}(t)|\, dt\leq C\rho^{n},\ n\in\mathbb{N},
\end{equation}
where $f^{*}$ is the invariant density of the Markov process $\tilde{X}_{n}$. 
Here $C>0$, $\rho<1$ are constants. 

For an arbitrary bounded, measurable function $\phi:\mathbb{R}\to\mathbb{R}$ one has almost sure convergence
\begin{equation}\label{lln1}
\frac{\phi(\tilde{X}_{1})+\ldots+\phi(\tilde{X}_{n})}{n}\to \int_{\mathbb{R}}\phi(x)f^{*}(x)\, dx,
\quad n\to\infty.
\end{equation}
Furthermore, one can approximate $f^{*}$ by the function iteration \eqref{eq:functioniteration} whose convergence rate is estimated by \eqref{esti}.
\end{theorem}

As an application of statement \eqref{lln1}, we determine the long-term relative frequency of the mispricing process being outside the
interval $[-\gamma_-,\gamma_+]$. (This quantity was denoted by $\mathtt{P}_{trade}$ in \cite{amm3}.)

\begin{corollary}\label{proba}
Almost surely, as $n\to\infty$,
$$
\frac{\sum_{j=1}^{n}1_{\{\tilde{X}_{j}\notin [-\gamma_{-},\gamma_{+}]\}}}{n}\to \mu_{*}(\mathbb{R}\setminus [-\gamma_{-},\gamma_{+}])=
\int_{\mathbb{R}\setminus [-\gamma_{-},\gamma_{+}]}f^{*}(t)\, dt.
$$	
\end{corollary}

Here $\mu_*$ denotes the invariant measure whose density is $f^*$, to which the system converges as time tends to infinity, see \eqref{esti} above.

In \cite{amm3} the arbitrage profit \emph{rate} was calculated, that is, the profit earned in an infinitesimally small interval.  
In the present discrete-time setting, we can also determine the arbitrageur's profit in one time step.
Assume now that the system is in stationary state (that is, $\tilde{X}_{n}$ has law $\mu_{*}$ for all $n$).

We consider a CPMM with total value locked \texttt{TVL} and we assume that 
$\gamma_{+}=\gamma_{-}=:\gamma$ (see \cite{amm1} for a more detailed discussion). 

\begin{corollary}\label{longterm}
For a CPMM in stationary state, the expected profit of the arbitrageur at any given instant can be expressed using the SPD $f^{*}$ as follows:
\begin{equation*}
\texttt{TVL} \cdot p\left[e^{\gamma/2} \int_{\gamma}^\infty \Bigl(\cosh\!\bigl(\tfrac{t-\gamma}{2}\bigr) - 1\Bigr)\, f^{*}(t)\, dt
+ e^{-\gamma/2} \int_{-\infty}^{-\gamma} \Bigl(\cosh\!\bigl(\tfrac{t+\gamma}{2}\bigr) - 1\Bigr) f^{*}(t)\, dt \right]. 
\end{equation*}
\end{corollary}

\begin{corollary}\label{longterm_volume}
For a CPMM in stationary state, the expected arbitrage trading volume at any given instant can be expressed as
\begin{equation*}
\texttt{TVL} \cdot p\left[\int_{\gamma}^\infty \left(e^{\gamma/2} - e^{-t/2+\gamma}\right) f^{*}(t)\, dt
+
\int_{-\infty}^{-\gamma} \left(e^{-\gamma/2} - e^{-t/2-\gamma}\right) f^{*}(t)\, dt \right] . 
\end{equation*}
\end{corollary}

\begin{remark}{\rm 
The parameter choice $p=1$ corresponds to the situation where arbitrageurs arrive at \emph{every} time instant. 
This is the discrete-time version of the ``fast-block regime'' treated in \cite{amm3}, our results above apply directly to this situation too.}
\end{remark}

\section{Numerical Results}
\label{sec:numericalresults}

To ascertain the validity of our mathematical model, we developed a numerical integrator in C++ to calculate the SPD of the misprice process for different model parameters\footnote{Source code: \url{https://gitlab.com/jtapolcai/sde_amm}}.
Downstream, the logarithmic misprice process, taken from~\cite{amm3}, is understood as:
\begin{equation}
     z_t \triangleq \log P_t / \tilde{P}_t, 
     \label{eq:logmispdef}
\end{equation}
where $\tilde{P}_t$ is the DEX price and $P_t$ is the external market price. Our C++ driver calculates the invariant probability distribution of this process.

Appendix \ref{sec:convergence} analyses the convergence properties of the stationary distribution. 
Our results indicate that after approximately 100 iterations, the function iteration yields curves that are nearly indistinguishable from the limiting distribution. 
Thus in the rest of the section the function iteration was run for 100 steps on a grid of 201 points. 

In Appendix \ref{sec:refrence_values}, we demonstrate that our model can reproduce the numerical values reported in Table 1. of \cite{amm3} by computing the SPD to arbitrary order and accuracy. 

Appendix \ref{sec:statisticaltests} analyses hourly log-returns of major cryptocurrencies (Table~\ref{table:hourlylogreturns}). 
The results show heavy tails, skewness, and uniformly low Kolmogorov–Smirnov p-values, confirming that Gaussian-based models such as Black–Scholes cannot capture the empirical dynamics. 
This motivates the need for richer model families, including jump-augmented processes.

\subsection{Fitting a model on empirical data}
\label{sec_fitting}

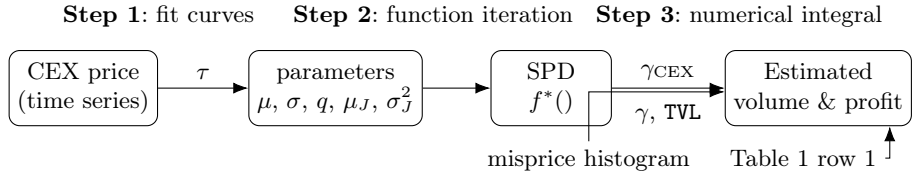
\begin{figure}
\begin{tikzpicture}[
  >={Latex[length=2mm]},
  node distance=2mm and 15mm,
  box/.style={draw, rounded corners, minimum width=16mm, minimum height=10mm, align=center},
  circ/.style={draw, circle, minimum size=8mm, inner sep=0pt, align=center},
  txt/.style={align=center},
  every label/.style={font=\footnotesize, inner sep=1pt},
 scale=0.5
]
\node[box] (mp) {CEX price\\(time series)};
\node[box, right=of mp, xshift=-3mm] (params) {parameters \\ $\mu$, $\sigma$, $q$, $\mu_J$, $\sigma_J^{2}$};
\node[box, right=of params, xshift=-6mm] (marg) {SPD \\ $f^*()$};
\node[box, right=of marg] (res) {Estimated \\ volume \& profit};
\draw[->] (mp) -- node[above,sloped,pos=.5]{$\tau$} (params); 
\node[above=of mp, xshift=10mm] {\textbf{Step 1}: fit curves};
\draw[->] (params) -- (marg);
\node[above=of params, xshift=13mm, align=center] {\textbf{Step 2}: function iteration};
\draw[->] (marg) -- node[above,sloped,pos=.5](l3){$\gammacex$} (res);
\node[below=of l3, yshift=1mm, align=center] {$\gamma$, \texttt{TVL}};
\node[above=of marg, xshift=25mm, align=center] {\textbf{Step 3}: numerical integral};
\node[below=of marg, align=center,  xshift=5mm] (hist) {misprice histogram};
\draw[->] (hist) |- 
([yshift=-1mm]res.west);
\node[below=of res, xshift=-2mm] (table1) {Table \ref{table:comp} row 1};
\draw[->] (table1) -| ([xshift=19mm]res.south);
\end{tikzpicture}
\caption{Flowchart of stochastic model evaluation on real data from an AMM.}
\label{fig:flowchart_measurement}
\end{figure}

To shed light on the applicability and the accuracy of our mathematical framework, we calculate the SPD of Eq.~\ref{eq:logmispdef} with and without jumps, see Fig \ref{fig:flowchart_measurement}.
We extract blockchain data using our own execution client, collecting transactions from the Uniswap V2 ETH--USDT AMM over the period 13--29 September 2025.

The values of $\tilde{P}_t$ are computed as follows. In parallel, we collect exchange-rate data at one-second resolution from multiple CEXs\footnote{Best ask price originate from Binance (84\%), KuCoin (9\%), and Gate.io (7\%), while best bids originate from Binance (88\%), KuCoin (5\%), and Gate.io (7\%).}. We assume competition among arbitrageurs: the arbitrageur who can source the asset at the best external price ultimately ``wins,'' as they can afford to pay the highest transaction fee and therefore have their transaction included on-chain. Accordingly, we compute the mispricing by comparing the best bid (or best ask) observed on the AMM to the best prices available across all monitored CEXs. We define a 12-second time window and assume that if, at any moment within this window, any CEX offers a better price, then some arbitrageur captures the opportunity and submits a transaction. We estimate that blocks arrive at our node with an approximately 8-second delay; the start of the 12-second window is aligned accordingly. The intuition behind this choice is that blocks are typically proposed within the first four seconds of a slot, while measuring this delay precisely would be extremely difficult. Finally, $\tilde{P}_t$ is inferred from the swap events’ amount-in and amount-out values.

Our stochastic model computes the expected number of arbitrages, their total volume, and the expected profit using seven parameters (inputs): TVL (\texttt{TVL}$=2L\sqrt{W}$), 
exchange fee ($\gamma$), daily price volatility ($\sigma$) and drift ($\mu$), and jump parameters — per-step jump probability ($q$), and jump-size distribution $U \sim \mathcal{N}(\mu_J,\sigma_J^{2})$ with mean $\mu_J$ and variance $\sigma_J^{2}$.

In Step 1, we estimate $\mu$, $\sigma$, $q$, $\mu_J$, and $\sigma_J^{2}$ from spot price data.
Recall that in our stochastic model, price movements arise from either a geometric Brownian motion (GBM) or discrete price jumps. In practice, these two sources cannot be perfectly distinguished. Therefore, a fixed jump-threshold parameter $\tau$ is commonly used: if the absolute return is smaller than $\tau$ times the price volatility, the movement is classified as GBM; otherwise, it is classified as a jump.
If $\tau$ is chosen too small, the jump fitting becomes inaccurate because many observations that actually belong to the GBM process are incorrectly classified as jumps. Conversely, if $\tau$ is chosen too large, too few jump observations remain for reliable parameter estimation. The measurements in the observed time period were not particularly sensitive to the choice of $\tau$; however, as a rule of thumb, we suggest setting $\tau = 2.0$. Under a Gaussian assumption, this threshold implies that approximately 4.55\% of GBM observations (two-sided) are incorrectly included in the jump sample, which represents a reasonable trade-off between bias and variance in the jump-parameter estimation. 

We then fit Gaussian curves to the empirical mispricing histograms: one to the GBM component, yielding estimates of $\mu$ and $\sigma$, and another to the jump component, providing $\mu_J$ and $\sigma_J$ for the jump-size distribution $U$ (these estimates were rescaled to match the time scale of our model).
The jump probability $q$ is inferred using the threshold $\tau$, which was set to attain higher log-likelihood against the observed time series. 

\begin{figure}[h]

 \fast{\includegraphics[width=\textwidth]{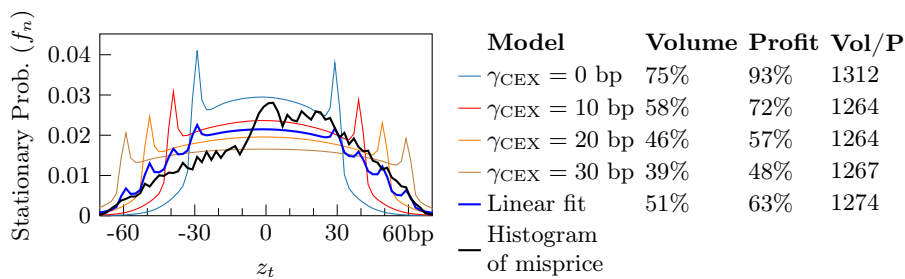}}
        {\begin{tikzpicture}
\definecolor{steelblue31119180}{RGB}{31,119,180}
\pgfplotsset{legend image post style={xscale=0.5}}

\begin{axis}[
  font=\footnotesize,
  height=40mm,
  legend columns=4,                
  legend cell align={left},
  legend style={
    fill opacity=1, row sep=2pt, 
    draw=none,
    at={(1.05,1.05)},
    anchor=north west,
    text opacity=1,
    align=left,
    font=\footnotesize,
    /tikz/nodes={inner sep=1pt},
    /tikz/column sep=1pt},
  ylabel style={yshift=-1.5ex},
  minor xtick={},
  minor ytick={},
  tick pos=left,
  width=60mm,
  x grid style={darkgrey176},
  xlabel={$z_t$},
  xmin=-0.007, xmax=0.007,
  xtick style={color=black},
  xtick={-0.006,-0.003,0,0.003,0.006},
  xticklabels={-60, -30, 0, 30, 60bp},
  scaled x ticks=false,
  minor x tick num=2,
  y grid style={darkgrey176},
  ylabel={Stationary Prob. ($f_n$)},
  ymin=0,
  ytick style={color=black},  
  ytick={0,0.01,0.02,0.03,0.04},
  yticklabels={0,0.01,0.02,0.03,0.04},
    scaled y ticks=false,
]

\addlegendimage{empty legend}\addlegendentry{\textbf{Model}}
\addlegendimage{empty legend}\addlegendentry{\textbf{Volume}}
\addlegendimage{empty legend}\addlegendentry{\textbf{Profit}}
\addlegendimage{empty legend}\addlegendentry{\textbf{Vol/P}}

\addplot[steelblue31119180,mark=none]
  table [x=x, y expr={\thisrow{density_0}/5000}, col sep=comma]
  {figures/mixture_fit_out.csv};
\addlegendentry{$\gamma_{\text{CEX}}=0$ bp}     
\addlegendimage{empty legend}\addlegendentry{75\%}  
\addlegendimage{empty legend}\addlegendentry{93\%}  
\addlegendimage{empty legend}\addlegendentry{1312}

\addplot[red,mark=none]
  table [x=x, y expr={\thisrow{density_1}/5000}, col sep=comma]
  {figures/mixture_fit_out.csv};
\addlegendentry{$\gamma_{\text{CEX}}=10$ bp}
\addlegendimage{empty legend}\addlegendentry{58\%}
\addlegendimage{empty legend}\addlegendentry{72\%}
\addlegendimage{empty legend}\addlegendentry{1264}

\addplot[orange,mark=none]
  table [x=x, y expr={\thisrow{density_2}/5000}, col sep=comma]
  {figures/mixture_fit_out.csv};
\addlegendentry{$\gamma_{\text{CEX}}=20$ bp}
\addlegendimage{empty legend}\addlegendentry{46\%}
\addlegendimage{empty legend}\addlegendentry{57\%}
\addlegendimage{empty legend}\addlegendentry{1264}
\addplot[brown,mark=none]
  table [x=x, y expr={\thisrow{density_3}/5000}, col sep=comma]
  {figures/mixture_fit_out.csv};
\addlegendentry{$\gamma_{\text{CEX}}=30$ bp}
\addlegendimage{empty legend}\addlegendentry{39\%}
\addlegendimage{empty legend}\addlegendentry{48\%}
\addlegendimage{empty legend}\addlegendentry{1267}
\addplot[thick,blue,mark=none]
  table [x=x, y expr={\thisrow{mixture_density}/5000}, col sep=comma]
  {figures/mixture_fit_out.csv};
\addlegendentry{Linear fit}
\addlegendimage{empty legend}\addlegendentry{51\%}
\addlegendimage{empty legend}\addlegendentry{63\%}
\addlegendimage{empty legend}\addlegendentry{1274}
\addplot[thick,black,mark=none]
  table [x=x, y expr={\thisrow{empirical_counts}/100000}, col sep=comma]
  {figures/mixture_fit_out.csv};
\addlegendentry{Histogram \\of misprice}
\addlegendimage{empty legend}\addlegendentry{ }
\addlegendimage{empty legend}\addlegendentry{ }

\end{axis}
\end{tikzpicture}}
\caption{Comparison of the SPD for five different values of $\gammacex$ with the empirical distribution. A linear combination fit was also introduced to approximate the empirical mispricing histogram, yielding weights $w_{0}=19\%$, $w_{10}=19\%$, $w_{20}=28\%$, and $w_{30}=34\%$. A histogram of the empirical mispricing is also shown. The legend further reports the share of the computed arbitrage volume relative to the empirically measured total pool volume, as well as the share of profit relative to the MEV-based estimate (see row 1 of Table~\ref{table:comp}).
The last column of the legend reports the volume-to-profit ratio, which is only marginally affected by $\gammacex$.
    \label{fig:misprice_densities}}
\end{figure}

With $\tau = 2.0$, the estimated daily volatility is $1.9$\%, the drift is $0.023$, and the jump probability is $q = 0.045$ per one-second step, with mean $-24.9$\% and standard deviation $15.7$\% in log-return units.
For comparison, when fitting the same parameters under a pure diffusion model (i.e., setting $q=0$), we obtained a daily volatility of $2.1$\% and a drift of $0.046$.

Next, in Step 2, we evaluate the function iteration taking a few minutes with 200 steps. 
Finally, in Step 3, we compute trade volumes using the formula in \Cref{longterm_volume}, profits using \Cref{longterm}, and transaction counts corresponding to the lower- and upper-tail integrals of the SPD.
To this end, we needed pool parameters (TVL and $\gamma$ are directly observable). 
In practice, however, when solving the SDE we opt to use effective $\gamma$ that also includes the CEX leg. This is denoted by $\gammacex$, i.e., $\gamma_{-} = \gamma_{+} = \gamma + \gammacex$. 
Note that this essentially corresponds to the CEX bid–ask spread, which is not precisely known, and, as Fig. \ref{fig:misprice_densities} shows, our calculations are highly sensitive to this value.

Fig. \ref{fig:misprice_densities} compares the observed misprice histogram with the stationary distribution $f$ obtained from the SDE under different values of the $\gammacex$ parameter.
In our model, $\gammacex$ is constant, which results in two sharp discontinuities of the distribution at $\gamma_{-}$ and $\gamma_{+}$.
In practice, however, the boundary is blurred.
This can be explained either by temporal variation in $\gammacex$, which shifts the effective cut-off points, or by differences across CEX exchange rates, whose superposition results in smoother curves. 
As a further extension, the mid-price SDE can be made more sophisticated with a spread process, likely mean-reverting (e.g., spreads widening in volatile periods and narrowing in calm markets).

To approximate this effect, we use a linear combination over spreads $\gammacex\in{0,10,20,30}$bp. 
Let $f(x\mid \gammacex=g)$ denote the stationary density at spread $g$. 
We choose nonnegative weights $w_g$ (summing to one) to obtain the best fit of the empirical histogram via the linear combination $f_{\text{comb}}(x)=\sum_{g} w_g\,f\!\left(x\mid \gammacex=g\right)$. 
The resulting weights $w_g$ are shown in the caption of Fig. \ref{fig:misprice_densities}. 
This modelling approach is consistent with the assumption that the CEX spread is random and can take multiple values. This interpretation aligns well with our view of arbitrageurs’ activity in the market: at any given moment, different participants have access to different CEX venues, each characterised by its own bid–ask spread level.


\subsection{Cross verification of the results}
\label{sec:cross-ver}

\begin{figure}
 \fast{\includegraphics[width=12cm]{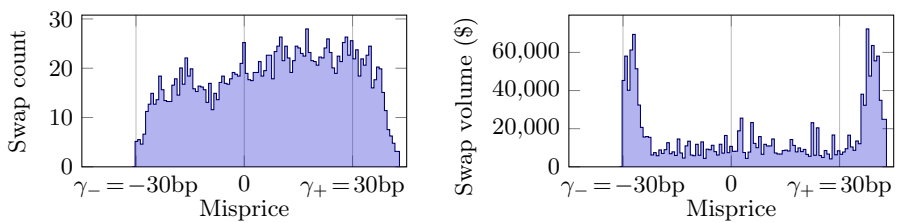}
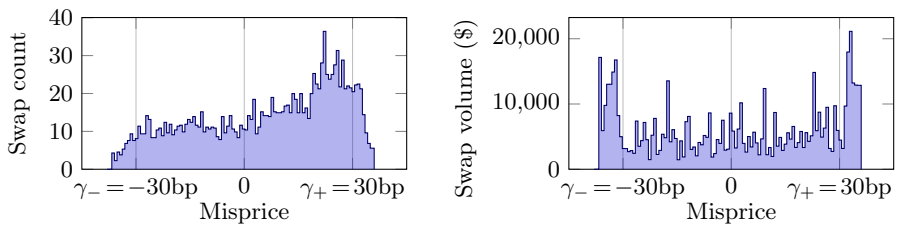}{
\tikzset{
  histplot/.style={
    const plot,
    draw=blue!40!black,
    fill=blue!80!black,
    fill opacity=0.30,
    mark=none
  }
}
\pgfplotsset{
  myaxis/.style={
    xlabel={Misprice},
    xmin=-0.45, xmax=0.45, ymin=0,
    xtick={-0.9,-0.6,-0.3,0,0.3,0.6,0.9},
    xticklabels={$-90$, $-60$, $\gamma_{-}\!=\!-30$bp, $0$, $\gamma_{+}\!=\!30$bp, $60$, $90$bp},
    xlabel style={yshift=2mm},xmajorgrids=true,
    scaled ticks=false,
    width=59mm, height=36mm
  }
}
\subfloat[
Daily swap count and volume histograms for Uniswap V2 ETH-USDT pool
\label{fig:volume_hist1}]{
\begin{tikzpicture}
\begin{axis}[myaxis, ylabel={Swap count},  ylabel style={yshift=-4mm}]
\addplot[histplot]
  table [x=binmidpointpercent, y=countperday, col sep=comma]
  {figures/misprice_histogram_data_0x0d4a.csv};
\end{axis}  
\end{tikzpicture}
\quad
\begin{tikzpicture}
\begin{axis}[myaxis, ylabel={Swap volume (\$)},  ylabel style={yshift=4pt,xshift=-2mm}]
\addplot[histplot]
  table [x=binmidpointpercent, y=volumeperdayusdc, col sep=comma]
  {figures/misprice_histogram_data_0x0d4a.csv};
\end{axis}  
\end{tikzpicture}
} \\
\subfloat[
Daily swap count and volume histograms for Uniswap V2 ETH-USDC pool
\label{fig:volume_hist2}]{
\begin{tikzpicture}
\begin{axis}[myaxis, ylabel={Swap count},  ylabel style={yshift=-4mm}]
\addplot[histplot]
  table [x=bin_midpoint_percent, y=count_per_day, col sep=comma]
  {figures/misprice_histogram_data_0xB4e16.csv};
\end{axis}  
\end{tikzpicture}
\quad
\begin{tikzpicture}
\begin{axis}[myaxis, ylabel={Swap volume (\$)},  ylabel style={yshift=4pt,xshift=-2mm}]
\addplot[histplot]
  table [x=bin_midpoint_percent, y=volume_per_day_usdc, col sep=comma]
  {figures/misprice_histogram_data_0xB4e16.csv};
\end{axis}  
\end{tikzpicture}
} }
\caption{Cross-validation measurements on two ETH–stablecoin Uniswap V2 pools confirming the small volume share of noise traders during September 13--29, 2025. 
For each swap, we computed the misprice before execution (CEX–DEX difference). 
The chart also shows the histogram, weighted by the trade volume. 
}
\label{fig:experiment_hist}
\end{figure}

The measurement in the previous section was fairly complex and relied on several assumptions; therefore, we conducted a second, much simpler measurement to verify the validity of our results. 

The first question we examined was whether the volume produced by noise traders is indeed smaller compared to the arbitrage volume. 
Fig.~\ref{fig:experiment_hist} presents the daily aggregate number of swaps (left) and their aggregate volume in USDT (right) as functions of the mispricing in Uniswap V2 pools (with $\gamma = 30$ bp).
We observe that trade counts are distributed uniformly across the $[-\gamma, \gamma]$ mispricing interval, but volumes are significantly higher near the boundaries, where arbitrage opportunities arise. 
The figures suggest that arbitrage accounts for approximately half of the volume and that $\gamma_{\mathrm{CEX}}$ is typically very small. Interestingly, the figure shows that arbitrage swaps can also occur for slightly negative $\gamma_{\mathrm{CEX}}$. This may be related to the fact that arbitrageurs may anticipate imminent price movements (as the examined time interval exhibits a mild price drift), or that they may receive information about price jumps slightly earlier, for example due to minor latencies in price data obtained via free CEX APIs.

Finally, we note that the arbitrage profit estimation in our model exhibits only limited sensitivity to $\gamma_{\mathrm{CEX}}$, suggesting that arbitrage volume can serve as a reasonably accurate proxy for arbitrage profits and the associated MEV.
 The last column of the legend on Fig. \ref{fig:misprice_densities} shows that the ratio of volume to profit computed from the SDE is not particularly sensitive to $\gammacex$. 
Intuitively, higher arbitrage profits necessarily require higher total volume. 
The resulting upper and lower bounds lie fairly close to each other, indicating that the MEV paid out from swaps is only slightly lower than the profit generated when the AMM is used exclusively by arbitrageurs in the examined AMM. 
Based on the analysis of price fluctuations in the ETH–USD currency pair, we conjecture that the arbitrage market is saturated: block builders can choose among many competing arbitrageurs, which forces them to rebate a large fraction of their profits.

\section{Conclusion}
\label{sec:conclusions}

In this work we developed an extended stochastic framework for modeling the price differences between CEXs and DEXs, treating the mispricing process as a Markov chain with state-dependent volatility, drift, and jumps. We proved that under mild assumptions the process admits a stationary distribution, and we provided constructive formulas—via function iteration—for calculating this invariant density. Building on this result, we derived closed-form expressions for arbitrage volumes and profits in a CPMM setting. Our estimates indicate that arbitrage profits are orders of magnitude larger than those predicted by earlier GBM-based models. 

Beyond theoretical analysis, we carried out extensive empirical work: we extracted DEX data from Ethereum, combined it with CEX price feeds, and subjected the resulting mispricing series to statistical tests. 
An open-source C++ implementation accompanies this research to aid reproducibility and further exploration. 

We evaluated the estimates of arbitrage volume derived from our model on the Uniswap V2 USDT–ETH pool and  found them to be of comparable magnitude, though slightly larger, than the measured CEX–DEX arbitrage volume.
This strengthens our confidence that the model’s assumptions are reasonably close to reality and that it can provide meaningful insights into several important practical questions, such as:
\begin{itemize}
  \item The revenue sustaining PBS arises naturally from the inherent dynamics of financial markets through CEX--DEX arbitrages in AMMs. 
  \item A large amount of AMM volume may be attributable to CEX--DEX arbitrage transactions. 
  \item CEX–DEX arbitrage profits are comparable to the measured MEV, suggesting that the arbitrage market may be saturated.        \item The heuristics employed in empirical studies to identify CEX--DEX arbitrage transactions may be overly restrictive. 
  \item Shorter block times are unlikely to reduce MEV, as arbitrage opportunities primarily emerge from rapid price jumps. 
\end{itemize}

\section*{Acknowledgements}
The project was partly supported by Project no. K23 146347 of National Research, Development and Innovation Fund of Hungary. 

\bibliographystyle{splncs04}
\bibliography{reference}

\appendix{}

\section{Proofs}

\begin{proof}[Proof of Lemma \ref{density1}] If $Z_{n+1}=0$ then the density of $\mu+\sigma\varepsilon_{n+1}$ is
$\frac{1}{\sigma}f\left(\frac{y-\mu}{\sigma}\right)$ as easily seen. If $Z_{n+1}=1$
then we need to calculate the convolution of $u(z)$ with $\frac{1}{\sigma}f\left(\frac{z-\mu}{\sigma}\right)$,{}
which equals
$$
\int_{\mathbb{R}}u(y-z)\frac{1}{\sigma}f\left(\frac{z-\mu}{\sigma}\right)
$$
for $y\in\mathbb{R}$. The statement now follows trivially.	
\end{proof}

\begin{proof}[Proof of Lemma \ref{density2}]
This follows directly from Lemma \ref{density1}
and from \eqref{tildex}, \eqref{tildey}.
\end{proof}

\begin{proof}[Proof of Corollary \ref{cor}]
Denoting by $\mu_{0}$ the law of $X_{0}$, we have
$$
P(X_{1}\in A)=\int_{A}\int_{\mathbb{R}}q(x_{0},x)\mu_{0}(dx_{0})dx,
$$	
so $X_{1}$ has density $f_{1}(x)=\int_{\mathbb{R}}q(x_{0},x)\mu_{0}(dx_{0})$. Repeating
the same argument, we can calculate $f_{k}(x)$, $x\in\mathbb{R}$, for all $k\geq 2$.
\end{proof}

We now state a general law of large numbers
for functionals of Markov processes, in a form that is convenient for us for the proof
of Theorem \ref{main}. For simplicity, we assume that the state
space is $\mathbb{R}$ but the result is true for
general state spaces. Let the Markov chain $Y_{n}$ have transition kernel $R(y,\cdot)$, that is, let
$$
P(Y_{n+1}\in A|Y_n=y)=R(y,A)
$$
hold for all Borel sets $A$ and for all 
$y\in\mathbb{R}$. For an arbitrary probability
$\nu$, the notation $R\nu$ refers to the
probability $A\to \int_\mathbb{R} R(y,A)\nu(dy)$.
Clearly, $R^n \nu$ (the $n$th iterate of the above
mapping) equals the law of $Y_n$ when the law of
$Y_0$ is $\nu$).

We recall the total variation distance of
two probabilities $\nu_1,\nu_2$:
$$
||\nu_1-\nu_2||_{TV}:=\sup_A|\nu_1(A)-\nu_2(A)|,
$$
where $A$ ranges over all the Borel sets in 
$\mathbb{R}$.

\begin{proposition}  Let the Markov chain $Y_n$ have a unique invariant measure $\nu_{*}$ 
and assume that, for all 
probabilities $\nu$ on $\mathbb{R}$, 
\begin{equation}\label{novo}
||R^{n}\nu-\nu_{*}||_{TV}\to 0,\ n\to\infty
\end{equation}
holds (that is, the law of $Y_{n}$ tends to $\nu_{*}$ whatever the initial distribution $\nu$ of $Y_{0}$ is). Then for all measurable
$\phi:\mathbb{R}\to\mathbb{R}$ such that $\int_{\mathbb{R}}|\phi(x)|\nu_{*}(dx)<\infty$, one has almost surely
\begin{equation}\label{lln}
\frac{\phi(Y_{1})+\ldots+\phi(Y_{n})}{n}\to\int_{\mathbb{R}}\phi(x)\nu_{*}(dx),\ n\to\infty.
\end{equation}
\end{proposition}
\begin{proof}
We rely on \cite{eric} and all the cited results are from that book. Note first that
\eqref{novo} implies that, for arbitrary probabilities
$\nu_1,\nu_2$, one has 
\begin{equation}\label{pori}
||R^n\nu_1-R^n\nu_2||_{TV}\to 0
\end{equation}
as $n\to\infty$.

By Proposition 11.A.3, \eqref{pori} implies that bounded harmonic functions are constant, so Corollary 5.2.4 ensures that the
tail sigma-field is trivial. As the invariant probability is unique, Theorems 5.2.6 and 5.2.1 
imply that the law of large numbers \eqref{lln} holds in the stated form. 
\end{proof}

The next technical lemma is also
necessary for the proof of Theorem \ref{main}.

\begin{lemma}\label{lowerestimate}
$$
\delta:=\inf_{(x_{0},x_{1})\in \mathbb{R}\times [-\gamma_{-},\gamma_{+}]}q(x_{0},x_{1})>0.
$$	
\end{lemma}
\begin{proof}
We have that
\begin{eqnarray*}
q(x,y)&\geq& p (1-q)\inf_{x\in\mathbb{R},y\in [-\gamma_-,\gamma_+]}\frac{1}{\sigma} 
f\left(\frac{y-\mu-x}
{\sigma}\right)\\
&\geq& \frac{p(1-q)}{K}
\inf_{x,y\in [-\gamma_-,\gamma_+]}f\left(\frac{y-\mu-x}{\sigma}\right)\geq  \frac{p(1-q)}{K}c_1
\end{eqnarray*}
\end{proof}

\begin{proof}[Proof of Theorem \ref{main}]
We rely on Chapter 16 of \cite{meyn-tweedie}. We define $\nu_{1}(A):=\mathrm{Leb}(A\cap [-\gamma_{-},\gamma_{+}])/(\gamma_{+}-\gamma_{-})$
and first verify that the whole state space of the process $\tilde{X}$ (that is, $\mathbb{R}$) 
is a $\nu_{1}$-small set (see Section 5.2 of \cite{meyn-tweedie} for the definition). 
Indeed, we have $q(x_{0},x_{1})\geq \delta>0$ for all $x_{0}\in\mathbb{R}$ and $x_{1}\in [-\gamma_{-},\gamma_{+}]$,
hence
$$
\inf_{x}Q(x,A)\geq \delta\nu_{1}(A)(\gamma_{+}-\gamma_{-}).
$$

Theorem 16.0.2 of \cite{meyn-tweedie} now implies that the Markov chain $\tilde{X}$ is \emph{uniformly ergodic}:
there is an invariant measure $\mu_{*}$ such that
\begin{equation}\label{tovo}
||Q^{n}(x,A)-\mu_{*}||_{TV}\leq C'\rho^{n}
\end{equation}
for suitable $C'>0$ and $\rho<1$. By invariance and Fubini's theorem,
$$
\mu_{*}(A)=\int_{\mathbb{R}}\int_{A}q(x_{0},x_{1})dx_{0}\mu_{*}(dx_{1})=\int_{A}\int_{\mathbb{R}}q(x_{0},x_{1})\mu_{*}(dx_{1})dx_{0},
$$ 
showing that $\mu_{*}$ is absolutely continuous w.r.t.\ the Lebesgue measure and hence admits a density $f_{*}$.
It follows from \eqref{tovo} that
\begin{equation}
\int_{\mathbb{R}}|f^{*}(t)-f_{n}(t)|\, dt\leq 2C'\rho^{n}.
\end{equation}
\end{proof}

\begin{proof}[Proof of Corollary \ref{proba}]
Apply \eqref{lln1} with the choice
$\phi(x):=1_{\{x\notin [-\gamma_-,\gamma_+\}}$.
\end{proof}

\begin{proof}[Proof of Corollary \ref{longterm}]
It is explained and calculated in Lemma 2 and Corollary 2 of \cite{amm3} (see also Example 2 of \cite{amm1}) that, if the current price is
$W$ and the mispricing is $z$ then the arbitrageur (if present) makes an amount	
$$
L\sqrt{W} e^{\gamma/2}[e^{(z-\gamma)/2} -2 + e^{-(z-\gamma)/2}]1_{\{z>\gamma\}}
$$
when mispricing is too high ($z>\gamma$) and
$$
L\sqrt{W}e^{-\gamma/2}[e^{(z+\gamma)/2} -2 + e^{-(z+\gamma)/2}]1_{\{z<-\gamma\}}
$$
when mispricing is too low ($z<-\gamma$). In a stationary state the distribution of $z$ is $\mu_{*}$.
The presence of an arbitrageur occurs with probability $p$, independently of what happened before.
Consider the second term without $L\sqrt{W}$ 
\[
p\int_{-\infty}^{-\gamma} e^{-\gamma/2}
\Bigl[e^{(t+\gamma)/2}-2+e^{-(t+\gamma)/2}\Bigr] f^{*}(t)\,dt.
\]
Using the identity
\[
e^{x}+e^{-x} = 2\cosh(x),
\]
the expression in brackets can be rewritten as
\[
e^{(t+\gamma)/2}-2+e^{-(t+\gamma)/2}
=2\cosh\!\left(\tfrac{t+\gamma}{2}\right)-2.
\]
Then we have $\texttt{TVL}=2L\sqrt(W)$, and the claimed formula now follows directly.
\end{proof}

\begin{proof}[Proof of Corollary \ref{longterm_volume}]
By Lemma~2 and Corollary~2 of \cite{amm3} (see also Example~2 of \cite{amm1}),
when the mid price is $W$ and the log mispricing is $t$, the arbitrageur (if present)
trades exactly the amount needed to bring the pool price back to the closest boundary
$e^{\pm\gamma}$ of the no-arbitrage band. 

For a CPMM with invariant $xy=L^2$, this rebalancing induces the following
\emph{gross notional volume} (denominated in the quote asset): 
\[
L\sqrt{W}\,p\left(e^{\gamma/2}-e^{-t/2+\gamma}\right)1_{\{t>\gamma\}}
\quad\text{and}\quad
L\sqrt{W}\,p\left(e^{-\gamma/2}-e^{-t/2-\gamma}\right)1_{\{t<-\gamma\}},
\]
respectively for the overpricing ($t>\gamma$) and underpricing ($t<-\gamma$) cases.
These follow from the CPMM algebra in \cite{amm3}, specialised to trade size
(gross volume) rather than profit. 

In the stationary regime, the distribution of $t$ is $\mu_{*}$ with density $f^{*}$, and the
arbitrageur is present independently with probability $p$.
Taking expectations with respect to $f^{*}$ yields the claimed expression.
\end{proof}

\begin{remark}\label{hollo}{\rm 
For the general class of constant function market makers (CFMMs), the trading function 
$f(x,y)=x^\theta y^{1-\theta}$ 
for some $\theta\in (0,1)$ and
positions in the risky/riskless asset must 
satisfy $f(x,y)=L$ at all time. 
As it is calculated in the second example of~\cite{amm1}, the holdings maximising pool value, at price $W$, are equal to
\begin{equation}\label{forma}
x^*(W)=L\left(\frac{\theta}{(1-\theta)W}\right)^{1-\theta},\quad y^*(W)=L\left(\frac{(1-\theta)W}{\theta}\right)^{\theta}.
\end{equation}
Let us now introduce the constants 
$$
c_1=
\left(\frac{\theta}{1-\theta}\right)^{1-\theta},
\quad c_2=
\left(\frac{1-\theta}{\theta}\right)^\theta.
$$
Calculations analogous to those of 
Corollary \ref{longterm} show that, in this case
the expected profit of the arbitrageur at a given 
trading date equals
\begin{eqnarray}
& & L{W}^\theta p\int_{\gamma}^\infty [c_1(e^{t(1-\theta)}-e^{\gamma(1-\theta)}) + c_2
(e^{\gamma-t\theta}-e^{\gamma(1-\theta)})]
f^{*}(t)\, dt \label{wq:arb_profit1}\\
&+& L{W}^\theta p\int_{-\infty}^{-\gamma}[
c_1(e^{t(1-\theta)}-e^{\gamma(\theta-1)})+
c_2(e^{-t\theta-\gamma}-e^{\gamma(\theta-1)})]f^{*}(t)\, dt \enspace . \label{wq:arb_profit2}
\end{eqnarray}

}
\end{remark}

\section{Limitations}
\label{sec:limitations}

Notwithstanding the fact that an attempt has been made to derive mathematical methods for a larger class of stochastic models to arrive at a more complete picture, our approach has some limitations. In this work, only L1 blockchains are considered and we believe that our discussion can be extended to L2 chains, whereby transactions are executed in rollups. Such an approach would entail much lower transaction fees. Another limitation of our work is that we only investigate V2 AMMs and no analysis has been made on concentrated liquidity provision (Uniswap V3). Finally, let us draw attention to the fact that the stochastic effects of noise traders are also neglected in this work. This topic shall constitute new research directions.

\section{Related Works}
\label{sec:related}

The similarities between market scoring rules and automated market making were pointed out by R. Hanson~\cite{hanson2007logarithmic} as early as 2002. The performance of a rebalancing strategy for various AMMs has been investigated in~\cite{amm1}. The concept of LVR, along with empirical studies, is also present in this work. The authors of~\cite{amm3} set out to find a closed-form solution of the AMM problem when fees are present. A comprehensive framework is provided in~\cite{complexityapproximationtradeoffsinexchnagemechanisms} for describing and analysing exchange mechanisms, enabling straightforward comparisons between limit order books and AMMs in terms of complexity and expressiveness. 
The work elucidates the notion of exchange complexity, clarifying the extent to which certain exchange methods are simpler compared to others. In~\cite{AMyersonianFrameworkforOptimalLiquidity}, the authors analyse the profit-maximising strategy of a monopolist liquidity provider AMMs using a Bayesian-inspired belief framework. 
The authors in~\cite{the_costs_of_swapping_on_the_uniswap_protocol} evaluate the efficiency of DEX trading using slippage as a key metric; their study explores the components of transaction costs, highlighting the variation based on trade characteristics.~\cite{rao2023trianglefees} investigates new fee structures to make AMMs more robust. 
In another line of work, FLAIR~\cite{milionis2023flair} has been introduced as a novel metric for liquidity provider competitiveness, complementing LVR.
\cite{Dewey2023CFMM} bears resemblance to our work, notably through its empirical investigations and model extensions, however, unlike our work the authors present no rigorous mathematical discussion on the existence of the SPD.~\cite{defiarbitrage} focuses on developing a risk-neutral pricing and hedging framework for constant product market makers (CPMM) liquidity tokens aiming to integrate this DeFi product into the principles of modern financial theory. 
A fairly recent work~\cite{Hasbrouck2023DecentralizedExchange} puts forth an economic model of a DEX with concentrated liquidity provision, examining a continuous-time model for a single DEX facilitating the trading of a risky asset. The authors demonstrate that the portfolio characteristics of concentrated liquidity provision can be directly linked to those of a covered call trading strategy. The findings of~\cite{conceptualflaws2021} suggest that when there are sufficiently many liquidity providers who are not overly risk-averse, linear pricing becomes much more cost-effective for investors.

\section{Convergence of the stationary distribution}
\label{sec:convergence}

\begin{figure}[h!]
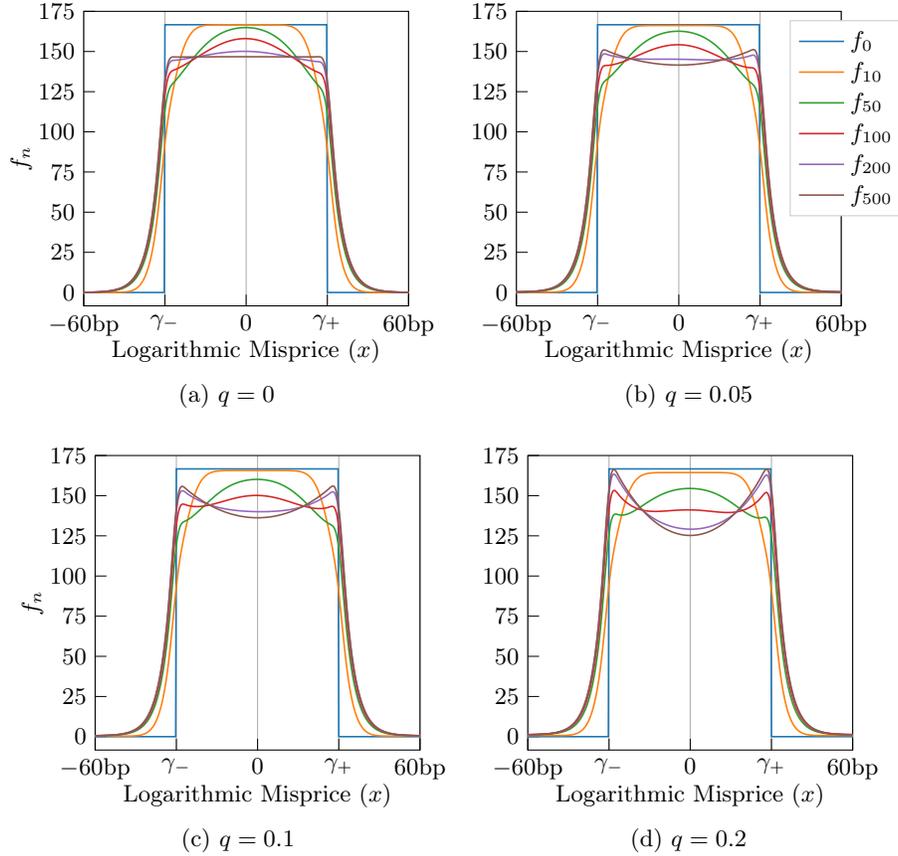

    \centering
    \fast{
        \includegraphics[width=\textwidth]{figures/function_iteration/f_functions_Q_series_NIT1000_NPT801.png}
}{
\pgfplotsset{
  sdpstyle/.style={
font=\small,
height=55mm, width=59mm,
minor xtick={},
minor ytick={},
tick pos=left,
x grid style={darkgrey176},
xlabel={Logarithmic Misprice (\(\displaystyle x\))},
xmin=-0.006, xmax=0.006,
xtick style={color=black},
xtick={-0.006,-0.004,-0.002,0,0.002,0.004,0.006},
y grid style={darkgrey176},
ylabel={}, ylabel style={yshift=-4mm},
ymin=-8.333333333335, ymax=175.000000000035,
ytick style={color=black},
ytick={0,25,50,75,100,125,150,175},
    xtick={-0.006,-0.003,0,0.003,0.006},
    xticklabels={$-60$bp, $\gamma_{-}$, $0$, $\gamma_{+}$, $60$bp},
    xlabel style={yshift=2mm},xmajorgrids=true,
    scaled ticks=false,
    legend cell align={left},
legend style={fill opacity=1, draw opacity=1, text opacity=1, draw=lightgrey204, at={(1.2,0.97)}},
  }
}
\subfloat[$q=0$]{\input{figures/fig_p0.tex}} \enspace
\subfloat[$q=0.05$]{\input{figures/fig_p05.tex}}
\\
\subfloat[$q=0.1$]{\input{figures/fig_p1.tex}} \enspace
\subfloat[$q=0.2$]{\input{figures/fig_p2.tex} 
}
}
    \caption{Convergence of the SPD of the misprice process with ($q > 0$) and without jumps ($q = 0$). Observe that the stationary state the system converges onto with different $q$ values. The number of grid points was set at 801. The density was iterated 1000 times and for clarity, only some of the functions are presented in the figures.}
    \label{fig:convergence}
\end{figure}
Now we turn to the investigation of the convergence properties of the SPD. The zeroth order density function \( f_0(x, \gamma_+, \gamma_-) \) is defined as a window function:  
\[
f_0(x, \gamma_+, \gamma_-) = 
\begin{cases} 
1 & \text{if } -\gamma_- \leq x \leq \gamma_+, \\
0 & \text{otherwise.}
\end{cases}
\]

Note that other initial guesses are also possible (e.g. Dirac delta in 0), however for faster convergence it is expedient to use an initial distribution that is as close to the final one as possible. The higher order densities are calculated using Eq.~\ref{eq:functioniteration} with a numerical integration formula. We normalise the density at each step and let the integrator run its course using the results of the previous iteration. The numerical integration based on Eq.~\ref{eq:functioniteration} was performed on a discrete grid using the Gauss–Kronrod quadrature formula with 15 quadrature points. 
The convergence of the stationary distribution of the misprice process was investigated using a numerical integrator. Our results are presented in Fig.~\ref{fig:convergence}. For this calculation, we used the same parameter set as in~\cite{amm3}, e.g. $\sigma = 5\%$, $\mu = \sigma^2/2$, $\Delta t = 12$s, $-\gamma = +\gamma = 0.003$ and the integration was performed on 801 grid points. In Fig.~\ref{fig:convergence} one can observe the effect of jumps (i.e. when $q > 0$). With jumps, the SPD has a well in the middle, implying that it is rare for the misprice to be around 0 compared to the case with no jumps. This is in accord with our intuition, as jumps in the external market price lead to increased arbitrage activity and the arbitrageurs tend to push the AMM back to the boundary more often. 

\section{Reproducing reference values}
\label{sec:refrence_values}

Our objective was first to reproduce the numerical values in Table 1. of~\cite{amm3} with our numerical solver that is capable of calculating the SPD up to arbitrary order and accuracy. The iteration order and the grid size are external parameters to the integrator module. 

The size of the grid has been varied and the absolute error for different block times and fee levels with respect to the reference values taken from Table 1. of~\cite{amm3} are presented in Fig.~\ref{fig:reference_errors} and the calculated values are presented in Table~\ref{tab:trade_region_comparison}. The calculation of the fast block region presented serious difficulties for the numerical integrator. Therefore the data points of $50$ msec have not been reproduced.

\begin{table}[htbp]
\centering
\caption{Comparison of reference trade region coverage (reference values are taken from~\cite{amm3}) with calculated results (in parentheses) for various block times ($\Delta t$) and fee levels ($\gamma$, expressed in basis points). Integrated results are rounded to one decimal place.}
\label{tab:trade_region_comparison}
\begin{tabular}{l|ccccc}
\hline
$\Delta t \setminus \gamma$ & 1 bp   & 5 bp   & 10 bp  & 30 bp  & 100 bp \\
\hline
10 min & 96.7\% (96.8\%) & 85.5\% (85.5\%) & 74.7\% (74.6\%) & 49.6\% (49.4\%) & 22.8\% (20.8\%) \\
2 min  & 92.9\% (93.1\%) & 72.5\% (72.6\%) & 56.9\% (56.9\%) & 30.5\% (30.5\%) & 11.6\% (11.6\%) \\
12 sec & 80.7\% (81.9\%) & 45.6\% (45.9\%) & 29.5\% (29.6\%) & 12.3\% (12.2\%) & 4.0\% (4.0\%)  \\
2 sec  & 63.0\% (67.2\%) & 25.4\% (26.2\%) & 14.5\% (14.8\%) & 5.4\% (5.4\%)  & 1.7\% (1.7\%)  \\
\hline
\end{tabular}
\end{table}

The model parameters are set according to that of~\cite{amm3} in the symmetric case. 
The daily volatility was set at 5\%, the diffusion drift ($\mu$) was set at $\sigma^2/2$ and the fee levels and block times were varied to reproduce the results of~\cite{amm3}. 
The size of the trade region was calculated on the discrete grid as the ratio of $\int^{-\gamma}_{-\infty}f^{*}(t)dt +\int^{+\infty}_{+\gamma}f^{*}(t)dt$ compared to $\int^{+\infty}_{-\infty}f^{*}(t)dt$.  
Note that the probability of stochastic jumps in this case was set to zero, i.e. $q =0$. At the end of each iteration, we normalise the densities, such that their integral from $-\infty$ to $\infty$ equals 1. When using our initial approach for the numerical integrator, we experienced numerical instabilities at the calculation of $h(x, y)$. We traced it to the term in which $\frac{1}{\sigma}$ is multiplied by $f\left(\frac{y-\mu(x)}{\sigma(x)}\right)$. When the volatility is small this operation leads to numerical instability, and we solved this issue by calculating the logarithm of these terms. The interested reader is referred to the open-source repository for the implementation details.
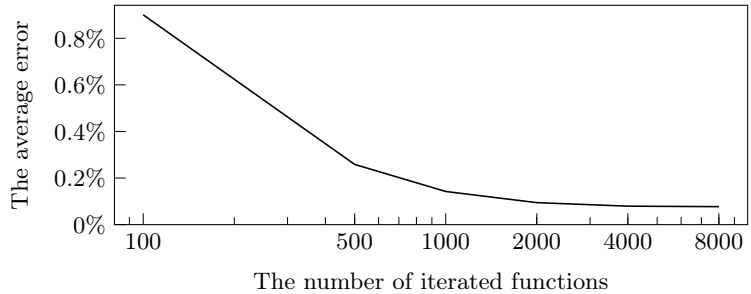
\begin{figure}[h!]
    \centering
\begin{tikzpicture}

\definecolor{darkgrey176}{RGB}{176,176,176}
\definecolor{lightgrey204}{RGB}{204,204,204}

\begin{axis}[
font=\small,
height=45mm,
legend cell align={left},
legend style={fill opacity=0.8, draw opacity=1, text opacity=1, draw=lightgrey204},
log basis x={10},
minor xtick={2,3,4,5,6,7,8,9,20,30,40,50,60,70,80,90,200,300,400,500,600,700,800,900,2000,3000,4000,5000,6000,7000,8000,9000,20000,30000,40000,50000,60000,70000,80000,90000,200000,300000,400000,500000,600000,700000,800000,900000},
minor ytick={},
tick pos=left,
width=100mm,
x grid style={darkgrey176},
xlabel={The number of iterated functions},
xmin=80.3240320295375, xmax=9959.65939192168,
xmode=log,
xtick style={color=black},
xtick={100,500,1000,2000,4000,8000},
xticklabels={100,500,1000,2000,4000,8000},
y grid style={darkgrey176},
ylabel={The average error},
ymin=0, ymax=0.941729988246625,
ytick style={color=black},
ytick={0,0.2,0.4,0.6,0.8,1},
yticklabels={$0$\%,$0.2$\%, $0.4$\%, $0.6$\%, $0.8$\%, $1$\%}
]
\addplot [semithick, black, mark=none]
table {%
100 0.900559806536667
500 0.258497412421667
1000 0.142323463159167
2000 0.0944325347041666
4000 0.0793160965766666
8000 0.0771561723375
};
\end{axis}

\end{tikzpicture}
    \caption{The average absolute errors of our numerical integrator with respect to the values of Table 1. of~\cite{amm3} with the number of function iterations. The number of grid points was set at 801. The average of errors with respect to the values in Table 1. of~\cite{amm3} is shown.}
    \label{fig:reference_errors}
\end{figure}
Notice that in some cases there is some residual error, but it can be decreased arbitrarily by increasing the number of grid points and the number of iterated functions. We observe the most discrepancy at $\Delta t = 2$ sec, $\gamma = 1$bp. For this data point the result converges very slowly to its reference value with an increased number of iterated functions and grid points. The reason behind this slow convergence is yet to be investigated.
Regarding the runtime of the numerical integrator, it calculates the data points that were necessary to generate Fig.~\ref{fig:reference_errors} in a few minutes fully utilising an AMD Ryzen 9 3900X 12-Core CPU.

\subsection{Arbitrage Profit}
\label{arb_profit_table}

Having calculated the SPD of the misprice process, we can calculate the arbitrage profits for a given trading interval. To this end, we rely on the above-introduced equation without the price ($W^{\theta}$) and the liquidity ($L$) terms, i.e.:
\begin{eqnarray}
\mathrm{ARB} = & & p\int_{\gamma}^\infty [c_1(e^{t(1-\theta)}-e^{\gamma(1-\theta)}) + c_2
(e^{\gamma-t\theta}-e^{\gamma(1-\theta)})]
f^{*}(t)\, dt\\
&+& p\int_{-\infty}^{-\gamma}[
c_1(e^{t(1-\theta)}-e^{\gamma(\theta-1)})+
c_2(e^{-t\theta-\gamma}-e^{\gamma(\theta-1)})]f^{*}(t)\, dt,
\label{eq:arbprofits}
\end{eqnarray}
where $p$ is the probability of an arbitrage trade, $\theta$ is a parameter of the AMM, $\gamma$ is the symmetric fee rate, $c_1$ and $c_2$ are as introduced above and $f^{*}(t)$ is the SPD of the logarithmic misprice process $z_t$. 

We calculate arbitrage profits in the symmetric setting, whereby $\mu = \sigma^2 /2$ and present our results in Table~\ref{tab:arbitrageprofits}. Notice that arbitrage profits increase significantly with increasing jump probability. This is in line with the common belief that arbitrage activity is generated by sudden price jumps. Observe for example the row where the block time is 12 seconds (Ethereum) and an AMM with 30 bp fee level (typical for Uniswap V2) and notice that jumps lead to orders of magnitude larger arbitrage profits.

\begin{table}[h!]
    \centering
    \caption{The proportion of $L\sqrt{W}$ accounting for arbitrage profits for different model parameters (block times and fee levels) without jumps and with jumps, with 4 different $q$ values. In each cell, the upper number corresponds to $q=0$ and the lower number corresponds to $q=0.2$. ARB was calculated according to Eq.~\ref{eq:arbprofits}. The number of grid points used was 801, the number of iterated distributions was set to 1000. $\theta = 0.5$, $\sigma = 5\%$ (daily), $\mu = \sigma^2 /2$.}
\begin{tabular}{llrrrrr}
\hline
$\Delta t$ & $q$ & $1$ bp & $5$ bp & $10$ bp & $30$ bp & $100$ bp \\
\hline
\hline
10 min & 0   & $\grayzero{0.00}202760$ & $\grayzero{0.00}175490$ & $\grayzero{0.00}149120$ & $\grayzero{0.00}087340$ & $\grayzero{0.00}023100$ \\
       & 0.05& $\grayzero{0.00}583060$ & $\grayzero{0.00}536270$ & $\grayzero{0.00}485070$ & $\grayzero{0.00}328570$ & $\grayzero{0.000}62900$ \\
       & 0.1 & $\grayzero{0.00}855130$ & $\grayzero{0.00}794330$ & $\grayzero{0.00}725320$ & $\grayzero{0.00}500800$ & $\grayzero{0.000}91520$ \\
       & 0.2 & $\grayzero{0.0}1199680$ & $\grayzero{0.0}1121090$ & $\grayzero{0.0}1029490$ & $\grayzero{0.00}718960$ & $\grayzero{0.00}129170$ \\
\hline
2 min  & 0   & $\grayzero{0.000}09700$ & $\grayzero{0.0000}7500$ & $\grayzero{0.0000}5900$ & $\grayzero{0.0000}3200$ & $\grayzero{0.0000}1100$ \\
       & 0.05& $\grayzero{0.000}38730$ & $\grayzero{0.000}34990$ & $\grayzero{0.000}31370$ & $\grayzero{0.000}21390$ & $\grayzero{0.0000}4600$ \\
       & 0.1 & $\grayzero{0.000}66180$ & $\grayzero{0.000}60930$ & $\grayzero{0.000}55420$ & $\grayzero{0.000}38560$ & $\grayzero{0.000}07800$ \\
       & 0.2 & $\grayzero{0.00}116750$ & $\grayzero{0.00}108700$ & $\grayzero{0.000}99710$ & $\grayzero{0.000}70100$ & $\grayzero{0.000}13500$ \\
\hline
12 sec & 0   & $\grayzero{0.000000}84$ & $\grayzero{0.000000}47$ & $\grayzero{0.000000}30$ & $\grayzero{0.000000}13$ & $\grayzero{0.0000000}4$ \\
       & 0.05& $\grayzero{0.00000}460$ & $\grayzero{0.00000}400$ & $\grayzero{0.00000}360$ & $\grayzero{0.00000}250$ & $\grayzero{0.000000}52$ \\
       & 0.1 & $\grayzero{0.00000}860$ & $\grayzero{0.00000}780$ & $\grayzero{0.00000}710$ & $\grayzero{0.00000}500$ & $\grayzero{0.00000}100$ \\
       & 0.2 & $\grayzero{0.0000}1800$ & $\grayzero{0.0000}1700$ & $\grayzero{0.0000}1500$ & $\grayzero{0.0000}1100$ & $\grayzero{0.00000}220$ \\
\hline
2 sec  & 0   & $\grayzero{0.0000000}2$ & $\grayzero{0.0000000}1$ & $\grayzero{0.00000000}$ & $\grayzero{0.00000000}$ & $\grayzero{0.00000000}$ \\
       & 0.05& $\grayzero{0.000000}13$ & $\grayzero{0.000000}11$ & $\grayzero{0.000000}10$ & $\grayzero{0.0000000}7$ & $\grayzero{0.0000000}1$ \\
       & 0.1 & $\grayzero{0.000000}25$ & $\grayzero{0.000000}22$ & $\grayzero{0.000000}20$ & $\grayzero{0.000000}15$ & $\grayzero{0.0000000}3$ \\
       & 0.2 & $\grayzero{0.000000}53$ & $\grayzero{0.000000}49$ & $\grayzero{0.000000}45$ & $\grayzero{0.000000}32$ & $\grayzero{0.0000000}7$ \\
\hline
\end{tabular}
\label{tab:arbitrageprofits}
\end{table}

\subsection{Further Results}

Using the blockchain data extracted from our execution client, we plot Figure~\ref{fig:daily_volatility} to demonstrate the daily volatility and the number of executed trades over an interval between November and December 2025. In the right subplot of this figure, a histogram of the exchanged amount of ETH is presented for the same interval. Notice that the number of trades per day can sometimes vary entirely independently of the daily volatility. This further supports the claim that volatility itself cannot be the only source of arbitrage trades on an AMM. 
\begin{figure}[h!]
    \centering
    \includegraphics[width=\textwidth]{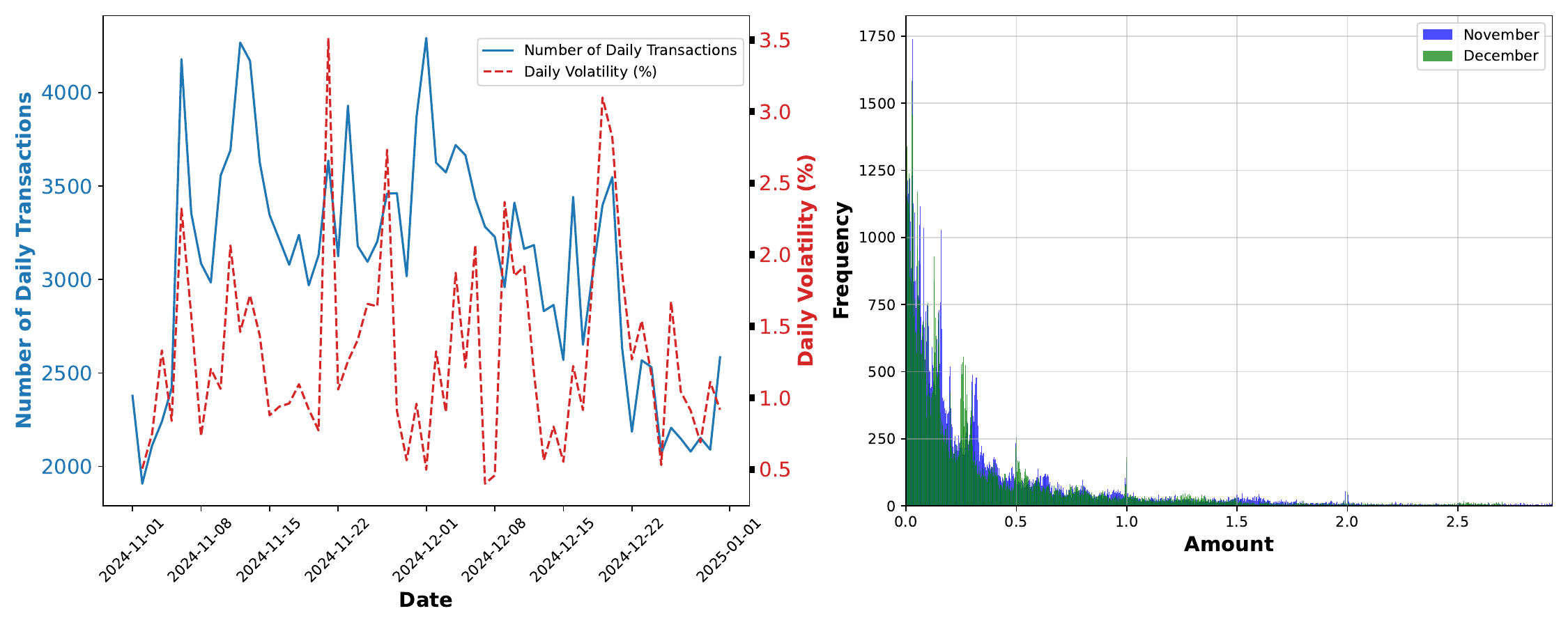} 
    \caption{The daily volatility in \% vs. the number of DEX transactions per day for each day throughout November and December 2024 (left) and a histogram of the amount of exchanged ETH during the investigated interval (right).}
    \label{fig:daily_volatility}
\end{figure}

\begin{table}
\caption{Estimated parameters of the fitted model for different jump thresholds $\tau$. 
Columns are the log-likelihood (LL), daily volatility parameter $\sigma$, daily mean parameter $\mu$, the per-step ($12$ sec) jump probability ($q$), and the daily jump-size distribution mean $\mu_U$ and daily variance $\sigma_U^{2}$.}
\label{table:fit_params}
\centering
\footnotesize
\setlength{\tabcolsep}{12pt} 
\begin{small}
\begin{tabular}{r | r r r r r r}
\hline
$\tau$ & LL & $\sigma$ & $\mu$ & $\mu_{J}$ & $\sigma_{J}$ & $q$\\
\hline
\hline
-- & 6.6192 & 0.0234 & 0.0234 & -- & -- & 0 \\
\hline
1.5 & 7.6895 & 0.0234 & 0.0517 & -0.0699 & 0.1649 & 0.0582 \\
1.6 & 7.6901 & 0.0234 & 0.0550 & -0.0911 & 0.1710 & 0.0541 \\
1.7 & 7.6903 & 0.0234 & 0.0587 & -0.1165 & 0.1767 & 0.0506 \\
1.8 & 7.6901 & 0.0234 & 0.0597 & -0.1321 & 0.1820 & 0.0476 \\
1.9 & 7.6896 & 0.0234 & 0.0598 & -0.1466 & 0.1883 & 0.0443 \\
2.0 & 7.6886 & 0.0234 & 0.0677 & -0.2035 & 0.1952 & 0.0411 \\
2.1 & 7.6875 & 0.0234 & 0.0724 & -0.2498 & 0.2018 & 0.0384 \\
2.2 & 7.6862 & 0.0234 & 0.0798 & -0.3135 & 0.2075 & 0.0362 \\
2.3 & 7.6847 & 0.0234 & 0.0791 & -0.3350 & 0.2138 & 0.0340 \\
2.4 & 7.6828 & 0.0234 & 0.0778 & -0.3571 & 0.2211 & 0.0317 \\
2.5 & 7.6809 & 0.0234 & 0.0799 & -0.4027 & 0.2281 & 0.0297 \\
3.0 & 7.6709 & 0.0234 & 0.0587 & -0.3345 & 0.2581 & 0.0228 \\
3.5 & 7.6603 & 0.0234 & 0.0470 & -0.2828 & 0.2868 & 0.0182 \\
4.0 & 7.6498 & 0.0234 & 0.0349 & -0.1623 & 0.3155 & 0.0148 \\
\hline
\end{tabular}
\end{small}
\end{table}

\section{Statistical Tests}
\label{sec:statisticaltests}

To underpin the applicability of the above-introduced mathematical framework, we investigated cryptocurrency exchange rates and applied statistical tests. Market data of cryptocurrencies \texttt{fil}, \texttt{btc}, \texttt{ETH}, \texttt{dxdy}, \texttt{xch}, \texttt{1inch}, \texttt{aave}, \texttt{ltc}, \texttt{beam}, \texttt{cro}, \texttt{sol} and \texttt{bnb} has been collected. Some of these cryptocurrencies are also tradable on platforms like Uniswap.
\begin{table}[ht]
\centering
\caption{Analysis results of hourly log-returns using the exchange rate of various cryptocurrencies from December 11, 2024, to January 15, 2025: kurtosis, skewness and Kolmogorov-Smirnov p-values.}
\begin{small}
\begin{tabular}{lrrr}
  \hline
Cryptocurrency & Kurtosis & Skewness & KSP \\ 
  \hline
1inch &  7.5280 & -0.3181 & 0.0027 \\ 
  aave &  5.7957 &  0.5147 & 0.0011 \\ 
  arb &  6.3933 & -0.1058 & 0.0021 \\ 
  beam & 14.4068 & -0.5742 & 0.0000 \\ 
  bnb &  9.6333 & -0.5604 & 0.0001 \\ 
  btc &  6.6794 & -0.4529 & 0.0000 \\ 
  cro & 13.7036 &  0.9348 & 0.0000 \\ 
  dxdy &  6.8077 &  0.2083 & 0.0008 \\ 
  ETH &  8.7076 & -0.5141 & 0.0000 \\ 
  fil &  7.5472 & -0.2352 & 0.0002 \\ 
  ltc & 14.7859 &  0.7149 & 0.0001 \\ 
  sol &  8.6612 & -0.1319 & 0.0001 \\ 
  xch &  7.8040 & -0.1468 & 0.0004 \\ 
   \hline
\end{tabular}
\label{table:hourlylogreturns}
\end{small}
\end{table}
In Table~\ref{table:hourlylogreturns}. results of a Kolmogorov-Smirnov test, along with kurtosis and skewness values are presented for these cryptocurrencies. Numerical data has been obtained from an API provided by GateIO\footnote{\href{https://www.gate.io/gate-api}{https://www.gate.io/gate-api}}. The high kurtosis values in the tables indicate that the corresponding distributions have fatter tails, meaning that they contain more extreme values or outliers compared to a normal distribution as conjectured in the introduction. The positive and the negative skewness values imply that the distributions are not centered, but rather stretched to the right/left. Finally, the small p-values for the Kolmogorov-Smirnov test for Gaussian distribution also implies that our conjectures were right.

This analysis on empirical data shows such heterogeneity in the invariant densities that it cannot be
reproduced by a simple two-parameter model family of price processes such as the Black-Scholes model.
We emphasise our findings again. In order to have a better match of the data with model predictions, a much larger class of
models should be considered. Section~\ref{sec:invariantprobability} of the present work guarantees that an invariant density
exists and can be calculated in a remarkably wide class of models.

\end{document}